\documentclass[11pt,a4paper]{article}

\usepackage{fullpage}
\usepackage[utf8x]{inputenc}
\usepackage{amsmath, amsthm}
\usepackage{wrapfig,graphicx,amssymb,textcomp,array,amsmath}
\usepackage{algpseudocode} 
\usepackage{enumitem}
\algtext*{EndWhile}
\algtext*{EndIf}
\usepackage{algorithm}
\usepackage{color}

\newcommand{\kTD}[2]{$#1\text{-{\em TD}}#2$}
\newcommand{\kDG}[2]{$#1\text{-}DG#2$}
\newcommand{\kGG}[2]{$#1\text{-}GG#2$}
\newcommand{\kRNG}[2]{$#1\text{-}RNG#2$}
\newcommand{\WS}[1]{\text{WS$(#1)$}}

\newcommand{\CD}[2]{D[#1,#2]}

\newcommand{\CIRC}[2]{C(#1,#2)}

\title{Matching in Gabriel Graphs\thanks{Research supported by NSERC.}}

\author{
Ahmad Biniaz\thanks{School of Computer Science, Carleton University, 
                    Ottawa, Canada.}
\and 
Anil Maheshwari\footnotemark[2]
\and 
Michiel Smid\footnotemark[2]
}
\date{\today}
\newtheorem{lemma}{Lemma}
\newtheorem{corollary}{Corollary}
\newtheorem{theorem}{Theorem}
\newtheorem{observation}{Observation}
\begin{document}

\maketitle

\begin{abstract}
Given a set $P$ of $n$ points in the plane, the order-$k$ Gabriel graph on $P$, denoted by \kGG{k}{}, has an edge between two points $p$ and $q$ if and only if the closed disk with diameter $pq$ contains at most $k$ points of $P$, excluding $p$ and $q$. We study matching problems in \kGG{k}{} graphs. We show that a Euclidean bottleneck perfect matching of $P$ is contained in \kGG{10}{}, but \kGG{8}{} may not have any Euclidean bottleneck perfect matching. In addition we show that \kGG{0}{} has a matching of size at least $\frac{n-1}{4}$ and this bound is tight. We also prove that \kGG{1}{} has a matching of size at least $\frac{2(n-1)}{5}$ and \kGG{2}{} has a perfect matching. Finally we consider the problem of blocking the edges of \kGG{k}{}.
\end{abstract}

\section{Introduction}
Let $P$ be a set of $n$ points in the plane. For any two points $p,q\in P$, let $\CD{p}{q}$ denote the closed disk which has the line segment $\overline{pq}$ as diameter. Let $|pq|$ be the Euclidean distance between $p$ and $q$.
The {\em Gabriel graph} on $P$, denoted by $GG(P)$, is defined to have an edge between two points $p$ and $q$ if $\CD{p}{q}$ is empty of points in $P\setminus\{p,q\}$. Let $C(p,q)$ denote the circle which has $\overline{pq}$ as diameter. Note that if there is a point of $P\setminus\{p,q\}$ on $C(p,q)$, then $(p,q)\notin GG(P)$. That is, $(p,q)$ is an edge of $GG(P)$ if and only if $$|pq|^2<|pr|^2+|rq|^2\quad\quad \forall r\in P,\quad\quad r\neq p,q.$$

Gabriel graphs were introduced by Gabriel and Sokal \cite{Gabriel1969} and can be computed in $O(n\log n)$ time \cite{Matula1980}. Every Gabriel graph has at most $3n-8$ edges, for $n\ge 5$, and this bound is tight \cite{Matula1980}. 

A {\em matching} in a graph $G$ is a set of edges without common vertices. A {\em perfect matching} is a matching which matches all the vertices of $G$. 
In the case that $G$ is an edge-weighted graph, a {\em bottleneck matching} is defined to be a perfect matching in $G$ in which the weight of the maximum-weight edge is minimized. For a perfect matching $M$, we denote the {\em bottleneck} of $M$, i.e., the length of the longest edge in $M$, by $\lambda(M)$. For a point set $P$, a {\em Euclidean bottleneck matching} is a perfect matching which minimizes the length of the longest edge. 

In this paper we consider perfect matching and bottleneck matching admissibility of higher order Gabriel Graphs. The {\em order-$k$ Gabriel graph} on $P$, denoted by \kGG{k}{}, is the geometric graph which has an edge between two points $p$ and $q$ iff $\CD{p}{q}$ contains at most $k$ points of $P\setminus\{p,q\}$. The standard Gabriel graph, $GG(P)$, corresponds to \kGG{0}{}. It is obvious that \kGG{0}{} is plane, but \kGG{k}{} may not be plane for $k\ge 1$. Su and Chang \cite{Su1990} showed that \kGG{k}{} can be constructed in $O(k^2n\log n)$ time and contains $O(k(n-k))$ edges. In \cite{Bose2013}, the authors proved that \kGG{k}{} is $(k+1)$-connected.
\subsection{Previous Work}
For any two points $p$ and $q$ in $P$, the {\em lune} of $p$ and $q$, denoted by $L(p,q)$, is defined as the intersection of the open disks of radius $|pq|$ centred at $p$ and $q$.
The {\em order-$k$ Relative Neighborhood Graph} on $P$, denoted by \kRNG{k}{}, is the geometric graph which has an edge $(p,q)$ iff $L(p,q)$ contains at most $k$ points of $P$.  
The {\em order-$k$ Delaunay Graph} on $P$, denoted by \kDG{k}{}, is the geometric graph which has an edge $(p,q)$ iff there exists a circle through $p$ and $q$ which contains at most $k$ points of $P$ in its interior. 
It is obvious that $$\text{\kRNG{k}{}}\subseteq\text{\kGG{k}{}}\subseteq\text{\kDG{k}{}}.$$

The problem of determining whether a geometric graph has a (bottleneck) perfect matching is quite of interest. Dillencourt showed that the Delaunay triangulation (\kDG{0}{}) admits a perfect matching \cite{Dillencourt1990}. Chang et al. \cite{Chang1992} proved that a Euclidean bottleneck perfect matching of $P$ is contained in \kRNG{16}{}.\footnote{They defined \kRNG{k}{} in such a way that $L(p,q)$ contains at most $k-1$ points of $P$.} This implies that \kGG{16}{} and \kDG{16}{} contain a (bottleneck) perfect matching of $P$. In \cite{Abellanas2009} the authors showed that \kGG{15}{} is Hamiltonian which implies that \kGG{15}{} has a perfect matching. 

Given a geometric graph $G(P)$ on a set $P$ of $n$ points, we say that a set $K$ of points {\em blocks} $G(P)$ if in $G(P\cup K)$ there is no edge connecting two points in $P$, in other words, $P$ is an independent set in $G(P\cup K)$.
Aichholzer et al.~\cite{Aichholzer2013} considered the problem of blocking the Delaunay triangulation (i.e. \kDG{0}{}) for $P$ in general position. They show that $\frac{3n}{2}$ points are sufficient to block DT($P$) and at least $n-1$ points are necessary. To block a Gabriel graph, $n-1$ points are sufficient, and $\frac{3}{4}n-o(n)$ points are sometimes necessary \cite{Aronov2013}.

In a companion paper \cite{Biniaz2014}, we considered the matching and blocking problems in triangular-distance Delaunay (TD-Delaunay) graphs. The {\em order-$k$ TD-Delaunay graph}, denoted by \kTD{k}{}, on a point set $P$ is the graph whose convex distance function is
defined by a fixed-oriented equilateral triangle. Then, $(p,q)$ is an edge in \kTD{k}{} if there exists an equilateral triangle which has $p$ and $q$ on its boundary and contains at most $k$ points of $P\setminus\{p,q\}$. We showed that \kTD{6}{} contains a bottleneck perfect matching and \kTD{5}{} may not have any. As for maximum matching, we proved that \kTD{1}{} has a matching of size at least $\frac{2(n-1)}{5}$ and \kTD{2}{} has a perfect matching (when $n$ is even). We also showed that $\lceil\frac{n-1}{2}\rceil$ points are necessary and $n-1$ points are sufficient to block \kTD{0}{}. In \cite{Babu2013} it is shown that \kTD{0}{} has a matching of size $\lceil\frac{n-1}{3}\rceil$. 

\subsection{Our Results}
In this paper we consider the following three problems: (a) for which values of $k$ does every \kGG{k}{} have a Euclidean bottleneck matching of $P$? (b) for a given value $k$, what is the size of a maximum matching in \kGG{k}{}? (c) how many points are sufficient/necessary to block a \kGG{k}{}? In Section~\ref{preliminaries} we review and prove some graph-theoretic notions. In Section~\ref{bottleneck-section} we consider the problem (a) and prove that a Euclidean bottleneck matching of $P$ is contained in \kGG{10}{}. In addition, we show that for some point sets, \kGG{8}{} does not have any Euclidean bottleneck matching. In Section~\ref{max-matching-section} we consider the problem (b) and give some lower bounds on the size of a maximum matching in \kGG{k}{}. We prove that \kGG{0}{} has a matching of size at least $\frac{n-1}{4}$, and this bound is tight. In addition we prove that \kGG{1}{} has a matching of size at least $\frac{2(n-1)}{5}$ and \kGG{2}{} has a perfect matching. In Section~\ref{blocking-section} we consider the problem (c). We show that at least $\lceil\frac{n-1}{3}\rceil$ points are necessary to block a Gabriel graph and this bound is tight. We also show that at least $\lceil\frac{(k+1)(n-1)}{3}\rceil$ points are necessary and $(k+1)(n-1)$ points are sufficient to block a \kGG{k}{}. The open problems and concluding remarks are presented in Section~\ref{conclusion}.

\section{Preliminaries}
\label{preliminaries}
Let $G$ be an edge-weighted graph with vertex set $V$ and weight function $w:E\rightarrow\mathbb{R^+}$. Let $T$ be a minimum spanning tree of $G$, and let $w(T)$ be the total weight of $T$. 

\begin{lemma}
\label{not-mst-edge}
Let $\delta(e)$ be a cycle in $G$ which contains an edge $e\in T$. Let $\delta'$ be the set of edges in $\delta(e)$ which do not belong to $T$ and let $e'_{max}$ be the largest edge in $\delta'$. Then, $w(e)\le w(e'_{max})$.
\end{lemma}

\begin{proof}
Let $e=(u,v)$ and let $T_u$ and $T_v$ be the two trees obtained by removing $e$ from $T$. Let $e'=(x,y)$ be an edge in $\delta'$ such that one of $x$ and $y$ belongs to $T_u$ and the other one belongs to $T_v$. By definition of $e'_{max}$, we have $w(e')\le w(e'_{max})$. Let $T'=T_u\cup T_v \cup\{(x,y)\}$. Clearly, $T'$ is a spanning tree of $G$. If $w(e')<w(e)$ then $w(T')<w(T)$; contradicting the minimality of $T$. Thus, $w(e)\le w(e')$, which completes the proof of the lemma. 
\end{proof}
 
For a graph $G=(V,E)$ and $S\subseteq V$, let $G-S$ be the subgraph obtained from $G$ by removing all vertices in $S$, and let $o(G-S)$ be the number of odd components in $G-S$, i.e., connected components with an odd number of vertices. The following theorem by Tutte~\cite{Tutte1947} gives a characterization of the graphs which have perfect matching: 

\begin{theorem}[Tutte~\cite{Tutte1947}] 
\label{Tutte} 
$G$ has a perfect matching if and only if $o(G-S)\le |S|$ for all $S\subseteq V$.
\end{theorem}

Berge~\cite{Berge1958} extended Tutte’s theorem to a formula (known as the Tutte-Berge formula) for the maximum size of a matching in a graph. In a graph $G$, the {\em deficiency}, $\text{def}_G(S)$, is $o(G-S)-|S|$. Let $\text{def}(G)=\max_{S\subseteq V}{\text{def}_G(S)}$.

\begin{theorem}[Tutte-Berge formula; Berge~\cite{Berge1958}] 
\label{Berge} 
The size of a maximum matching in $G$ is $$\frac{1}{2}(n-\mathrm{def}(G)).$$
\end{theorem}

For an edge-weighted graph $G$ we define the {\em weight sequence} of $G$, \WS{G}, as the sequence containing the weights of the edges of $G$ in non-increasing order. A graph $G_1$ is said to be less than a graph $G_2$ if \WS{G_1} is lexicographically smaller than \WS{G_2}.

\section{Euclidean Bottleneck Matching}
\label{bottleneck-section}
Given a point set $P$, in this section we prove that \kGG{10}{} contains a Euclidean bottleneck matching of $P$. We also present a configuration of a point set $P$ such that \kGG{8}{} does not contain any Euclidean bottleneck matching of $P$. We use a similar argument as in \cite{Abellanas2009, Chang1991}. First consider the following lemma of \cite{Abellanas2009}:
\begin{lemma}[Abellanas et al.~\cite{Abellanas2009}]
\label{cone-lemma}
Let $0 < \theta \le \pi/5$. Let $C(A, \theta, L, R)$ be a cone with apex $A$, bounding rays $L$
and $R$ emanating from $A$ and angle $\theta$ computed clockwise from $L$ to $R$. Given two points $x, y \in C(A, \theta, L, R)$ and a constant $r > 0$. If $|xA| > 2r$ and $|yA|>2r$, then $|xy| < 2r$ or $|xy| < \max\{|xA|− r, |yA|− r\}$.
\end{lemma}

\begin{figure}[htb]
  \centering
  \includegraphics[width=.6\columnwidth]{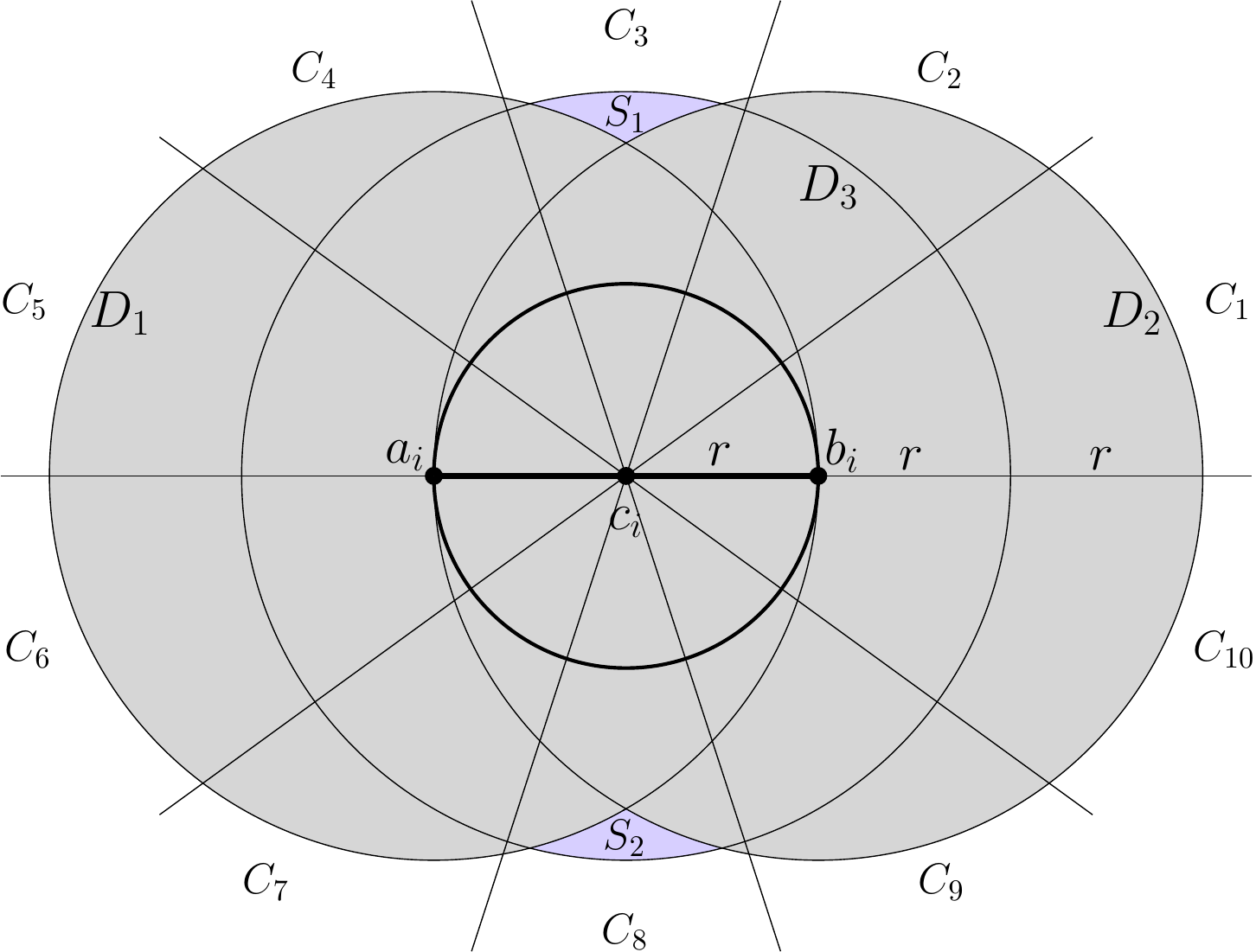}
 \caption{Illustration for Theorem~\ref{10-GG-thr}.}
  \label{10-GG-fig}
\end{figure}

\begin{theorem}
 \label{10-GG-thr}
For every point set $P$, \kGG{10}{} contains a Euclidean bottleneck matching of $P$.
\end{theorem}
\begin{proof}
Let $\mathcal{M}$ be the set of all perfect matchings through the points of $P$. Define a total order on the elements of $\mathcal{M}$ by their weight sequence. If two elements have exactly the same weight sequence, break ties arbitrarily to get a total order.
Let $M^* = \{(a_1, b_1),\dots, (a_{\frac{n}{2}}, b_{\frac{n}{2}})\}$ be a perfect matching in $\mathcal{M}$ with minimal weight sequence. It is obvious that $M^*$ is a Euclidean bottleneck matching for $P$. We will show that all edges of $M^*$ are in \kGG{10}{}. Consider any edge $e = (a_i, b_i)$ in $M^*$ and its corresponding disk $\CD{a_i}{b_i}$. Suppose that $\CD{a_i}{b_i}$ contains $w$ points of $P\setminus\{a_i,b_i\}$. Let $U = \{u_1, u_2,\dots, u_w\}$ represent the points inside $\CD{a_i}{b_i}$, and $U'=\{r_1, r_2,\dots, r_w\}$ represent the points where $(r_i,u_i)\in M^*$. We will show that $w\le 10$. Let $r=|a_ib_i|/2$ be the radius of $\CD{a_i}{b_i}$. 

{\em Claim 1}: For each $r_j\in U'$, $\min\{|r_ja_i|, |r_jb_i|\} \ge 2r$. To prove this, assume that $|r_ja_i|< 2r$ and let $M$ be the perfect matching obtained from $M^*$ by deleting $\{(a_i,b_i), (r_j,u_j)\}$, and adding $\{(a_i, r_j), (b_i, u_j)\}$. The two new edges are smaller than the old ones. Thus, $\WS{M}<\WS{M^*}$ which contradicts the minimality of $M^*$.

Let $D_1$ and $D_2$ respectively be the open disks with radius $2r$ centered at $a_i$ and $b_i$. By Claim 1, we may assume that no point of $U'$ lies inside $D_1\cup D_2$. In other words all points of $U'$ are contained in $\overline{D_1\cup D_2}$.

{\em Claim 2}: For each pair $r_j$ and $r_k$ of points in $U'$, $|r_jr_k|\ge \allowbreak\max\allowbreak\{|a_ib_i|, \allowbreak |r_ju_j|, \allowbreak |r_ku_k|\}$. To prove this, assume that $|r_jr_k|< \max\{|a_ib_i|, |u_jr_j|,\allowbreak |u_kr_k|\}$. Let $M$ be the perfect matching obtained from $M^*$ by deleting $\{(u_j,r_j),(u_k,r_k),(a_i, b_i)\}$ and adding $\{(a_i, u_j), (b_i, u_k),(r_j, r_k)\}$. Since $\max \{|a_iu_j|,\allowbreak |b_iu_k|,\allowbreak |r_jr_k|\}<\max\{|u_jr_j|,\allowbreak |u_kr_k|,\allowbreak |a_ib_i|\}$, $\WS{M} \allowbreak<\allowbreak \WS{M^*}$ which contradicts the minimality of $M^*$.

Let $c_i$ be the center of $\CD{a_i}{b_i}$. Consider a decomposition of the plane into 10 cones $C_1, \dots, C_{10}$ of angle $\pi/5$ with apex at $c_i$. See Figure~\ref{10-GG-fig}. By contradiction, we will show that each cone $C_i, 1\le i\le 10$, contains at most one point of $U'$. Suppose that a cone $C_i$ where $1\le i\le 10$ contains two points $r_j, r_k\in U'$. It is obvious that 

\begin{equation}\label{equ1}
|r_ju_j|\ge |c_ir_j|-r \quad\quad\text{and} \quad\quad |r_ku_k|\ge |c_ir_k|-r.
\end{equation}

{\em Claim 3}: Each cone $C_i$ where $1\le i\le 10$ and $i\neq 3,8$ contains at most one point of $U'$. Suppose that $C_i$ contain two points $r_j,r_k \in U'$. By Claim 1, all points of $U'$ are contained in $\overline{D_1\cup D_2}$. Consider the disk $D_3$ with radius $2r$ centred at $c_i$, as shown in Figure~\ref{10-GG-fig}. Since $D_3\cap \allowbreak(\overline{D_1\cup D_2})=\allowbreak\emptyset$, $r_j$ and $r_k$ are outside $D_3$, i.e., $|r_jc_i| > 2r$ and $|r_kc_i|>2r$. By Lemma~\ref{cone-lemma}, $|r_jr_k| < 2r$ or $|r_jr_k| \allowbreak<\allowbreak \max\{|r_jc_i| − r, \allowbreak |r_kc_i|− r\}$. By inequality~(\ref{equ1}), $|r_jr_k| < \max\{|a_ib_i|,|r_ju_j|,|r_ku_k|\}$ which contradicts Claim 2.

\begin{figure}[htb]
  \centering
\setlength{\tabcolsep}{0in}
  $\begin{tabular}{cc}
 \multicolumn{1}{m{.5\columnwidth}}{\centering\includegraphics[width=.27\columnwidth]{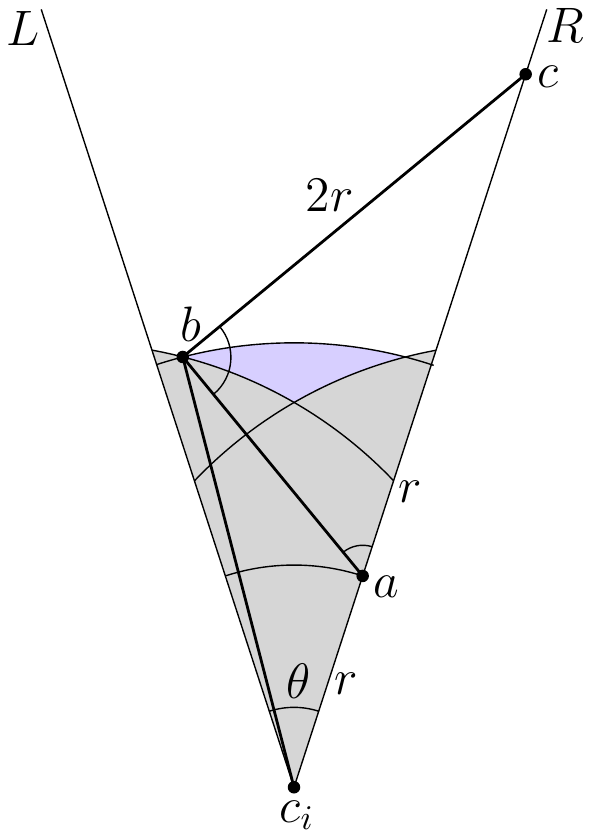}}
&\multicolumn{1}{m{.5\columnwidth}}{\centering\includegraphics[width=.3\columnwidth]{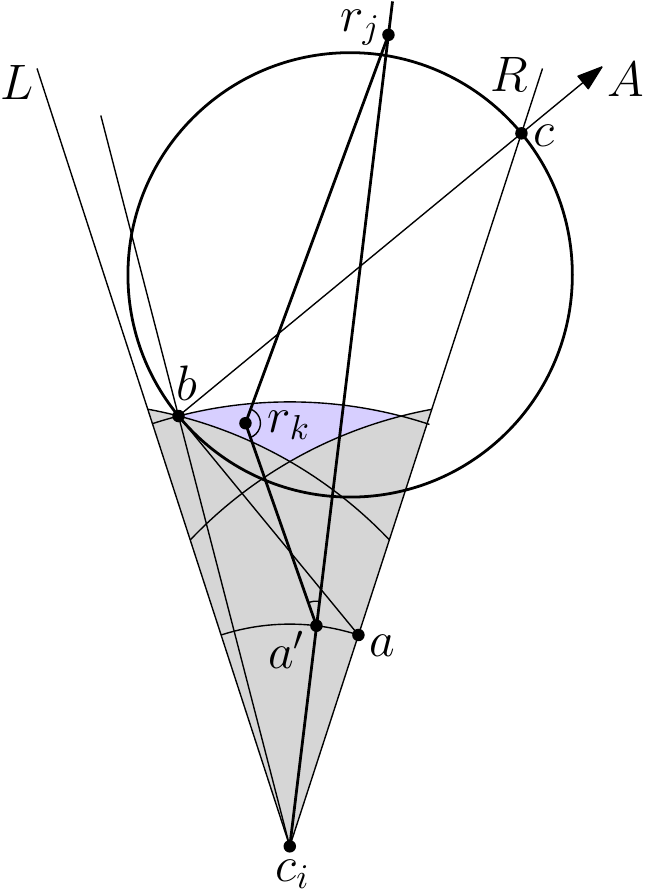}}\\
(a) & (b)
\end{tabular}$
  \caption{(a) The angle $\angle bac$ is smaller than the angle $\angle abc$, and hence (b) $\angle r_ka'r_j <\angle a'r_kr_j$.}
\label{cone-fig}
\end{figure}

{\em Claim 4}: Each of $C_3$ and $C_8$ contains at most one point of $U'$. Let $\{S_1,S_2\}$ be the partition of $D_3\cap (\overline{D_1\cup D_2})$ which lies inside $C_3$ and $C_8$ as shown in Figure~\ref{10-GG-fig}. Because of symmetry, we only prove the claim for $C_3$. Suppose that $C_3$ contains two points $r_j, r_k \in U'$.
For the rest of the proof, refer to Figure~\ref{cone-fig}. 
W.l.o.g. assume that $r_j$ is further from $c_i$ than $r_k$ and $r_k$ is to the left of $r_j$ (i.e., $r_k$ is to the left of the line through
$c_i$ and $r_j$ oriented from $c_i$ to $r_j$). If $r_k\notin S_1$ then $|r_kc_i|>2r$ and $|r_jc_i|>2r$. Then, by Lemma~\ref{cone-lemma} and Claim 2 we have a contradiction. Therefore, assume that $r_k\in S_1$. Let $L$ and $R$ be the two rays defining $C_3$. Let $a$ be the intersection of $R$ and $\CIRC{a_i}{b_i}$. Let $b$ be the intersection of the boundaries of $D_1$ and $D_3$ which is inside $C_3$. Define the point $c$ on $R$ such that $|bc|=2r$ and $c\neq c_i$. See Figure~\ref{cone-fig}(a). The triangle $\bigtriangleup cbc_i$ is isosceles, and hence $\angle bcc_i = \angle bc_ic <\frac{\pi}{5}$. This implies that $\angle cbc_i>\frac{3\pi}{5}$. On the other hand, in triangle $\bigtriangleup abc_i$, $|ab|>|ac_i|$, which implies that $\angle abc_i <\angle ac_ib<\frac{\pi}{5}$. Thus $\angle abc >\frac{2\pi}{5}$. In addition $\angle bac_i>\frac{3\pi}{5}$ and hence $\angle bac < \frac{2\pi}{5}$. Therefore in the triangle $\bigtriangleup abc$ we have $$\angle abc>\frac{2\pi}{5}>\angle bac.$$ Let $C(b,c)$ be the circle with radius $2r$ having $\overline{bc}$ as diameter, and let $A$ be the ray emanating from $b$ which goes through $c$ as shown in Figure~\ref{cone-fig}(b). The intersection of $C_3$ with $\overline{D_1\cup D_2}$ which lies to the right of $A$ is completely inside $C(b,c)$. Thus, if $r_j$ is to the right of $A$, $|r_jr_k|<\allowbreak 2r\allowbreak =\allowbreak |a_ib_i|$, which contradicts Claim 2. Therefore $r_j$ lies to the left of $A$. If $r_j$ is in the interior of $C_3$, rotate $C_3$ counter-clockwise around $c_i$ until $r_j$ lies on $R$. Since $r_k$ is to the left of $r_j$, the point $r_k$ is still in the interior of $C_3$. Let $a'$ be the intersection of the new $R$ with $\CIRC{a_i}{b_i}$. Note that $S_1$ and hence $r_k$ is contained in $\bigtriangleup abc$. In addition $r_j$ and $a'$ are outside $\bigtriangleup abc$ and to the left of the line through $a$ and $c$. Therefore, $\angle a'r_kr_j\ge \angle abc >\angle bac \ge r_ka'r_j$ and hence $$|r_jr_k|<|r_ja'|=|r_jc_i|-r\le |r_ju_j|,$$ which contradicts Claim 2. 

By Claim 3 and Claim 4 each cone $C_i$ where $1\le i \le 10$ contains at most one point of $U'$. Thus, $w\le 10$, and $e=(a_i,b_i)$ is an edge of \kGG{10}{}.
\end{proof}

\begin{figure}[htb]
  \centering
  \includegraphics[width=.8\columnwidth]{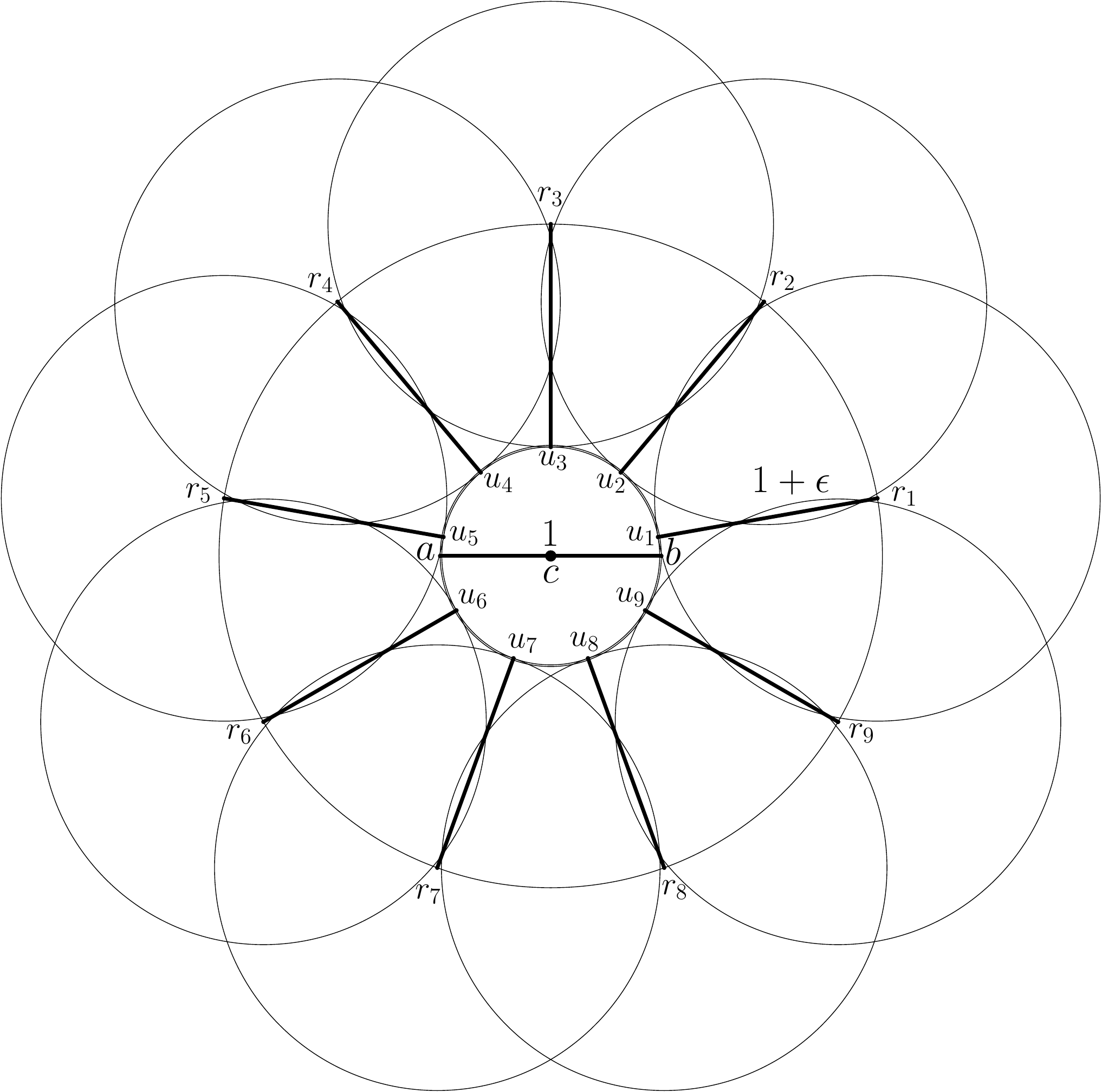}
 \caption{A set of 20 points such that \kGG{8}{} does not contain any Euclidean bottleneck matching.}
  \label{8-GG-fig}
\end{figure}
Now, we will show that for some point sets, \kGG{8}{} does not contain any Euclidean bottleneck matching.
Consider Figure~\ref{8-GG-fig} which shows a configuration of a set $P$ of 20 points. The closed disk $\CD{a}{b}$ is centred at $c$ and has diameter one, i.e., $|ab|=1$. $\CD{a}{b}$ contains 9 points $U=\{u_1, \dots, u_9\}$ which lie on a circle with radius $\frac{1}{2}-\epsilon$ which is centred at $c$. Nine points in $U'=\{r_1,\dots,r_9\}$ are placed on a circle with radius 1.5 which is centred at $c$ in such a way that $|r_ju_j|= 1+\epsilon$, $|r_ja|>1+\epsilon$, $|r_jb|>1+\epsilon$, and $|r_jr_k|>1+\epsilon$ for $1\le j, k\le 9$ and $j\neq k$. Consider a perfect matching $M=\{(a,b)\}\cup \{(r_i, u_i): i=1,\dots, 9\}$ where each point $r_i\in U'$ is matched to its closest point $u_i$. It is obvious that $\lambda(M)=1+\epsilon$, and hence the bottleneck of any bottleneck perfect matching is at most $1+\epsilon$. We will show that any Euclidean bottleneck matching of $P$ contains $(a,b)$. By contradiction, let $M^*$ be a Euclidean bottleneck matching which does not contain $(a,b)$. In $M^*$, $a$ is matched to a point $x\in U\cup U'$. If $x \in U'$, then $|ax|>1+\epsilon$. If $x\in U$, w.l.o.g. assume that $x = u_1$. Thus, in $M^*$ the point $r_1$ is matched to a point $y$ where $y\neq u_1$. Since $u_1$ is the closest point to $r_1$ and $|r_1u_1|=1+\epsilon$, $|r_1y|>1+\epsilon$. In both cases $\lambda(M^*)> 1+\epsilon$, which is a contradiction. Therefore, $M^*$ contains $(a,b)$. Since $\CD{a}{b}$ contains 9 points of $P\setminus\{a,b\}$, $(a,b)\notin\text{\kGG{8}{}}$. Therefore \kGG{8}{} does not contain any Euclidean bottleneck matching of $P$.

\section{Maximum Matching}
\label{max-matching-section}
Let $P$ be a set of $n$ points in the plane. In this section we will prove that \kGG{0}{} has a matching of size at least $\frac{n-1}{4}$; this bound is tight. We also prove that \kGG{1}{} has a matching of size at least $\frac{2(n-1)}{5}$ and \kGG{2}{} has a perfect matching (when $n$ is even).

First we give a lower bound on the number of components that result after removing a set $S$ of vertices from \kGG{k}{}. Then we use Theorem~\ref{Tutte} and Theorem~\ref{Berge}, respectively presented by Tutte~\cite{Tutte1947} and Berge~\cite{Berge1958}, to prove a lower bound on the size of a maximum matching in \kGG{k}{}. 

\begin{figure}[htb]
  \centering
\setlength{\tabcolsep}{0in}
  $\begin{tabular}{cc}
 \multicolumn{1}{m{.5\columnwidth}}{\centering\includegraphics[width=.38\columnwidth]{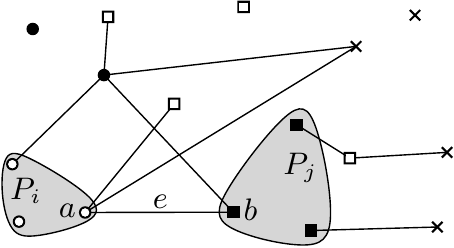}}
&\multicolumn{1}{m{.5\columnwidth}}{\centering\includegraphics[width=.4\columnwidth]{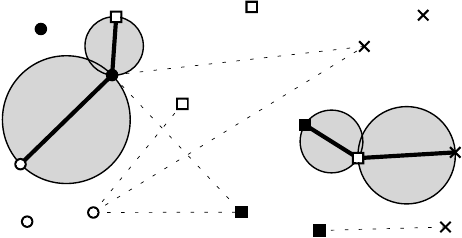}}\\
(a) & (b)
\end{tabular}$
  \caption{The point set $P$ of 16 points is partitioned into open/closed disks, open/closed squares, and crosses. (a) The graph $G(\mathcal{P})$, (b) The set $\mathcal{T}$ of straight-line edges corresponding to $MST(G(\mathcal{P}))$ is in bold, and the set $\mathcal{D}$ of their corresponding disks.}
\label{partition-fig}
\end{figure}

Let $\mathcal{P}=\{P_1, P_2,\dots\}$ be a partition of the points in $P$. For two sets $P_i$ and $P_j$ in $\mathcal{P}$ define the distance $d(P_i,P_j)$ as the smallest Euclidean distance between a point in $P_i$ and a point in $P_j$, i.e., $d(P_i,P_j)=\min\{|ab|:a\in P_i, b\in P_j\}$.
Let $G(\mathcal{P})$ be the complete edge-weighted graph with vertex set $\mathcal{P}$. For each edge $e=(P_i,P_j)$ in $G(\mathcal{P})$, let $w(e)=d(P_i,P_j)$. This edge $e$ is defined by two points $a$ and $b$, where $a\in P_i$ and $b\in P_j$. Therefore, an edge $e\in G(\mathcal{P})$ corresponds to a straight line edge $(a,b)$ in $P$; see Figure~\ref{partition-fig}(a). Let $MST(G(\mathcal{P}))$ be a minimum spanning tree of $G(\mathcal{P})$. It is obvious that each edge $e$ in $MST(G(\mathcal{P}))$ corresponds to a straight line edge $(a,b)$ in $P$. Let $\mathcal{T}$ be the set of all these straight line edges. Let $\mathcal{D}$ be the set of disks which have the edges of $\mathcal{T}$ as diameter, i.e., $\mathcal{D}=\{D[a,b]: (a,b)\in \mathcal{T}\}$. See Figure~\ref{partition-fig}(b).

\begin{observation}
 \label{T-plane}
$\mathcal{T}$ is a subgraph of a minimum spanning tree of $P$, and hence $\mathcal{T}$ is plane.
\end{observation}

\begin{lemma}
 \label{D-empty}
A disk $\CD{a}{b}\in\mathcal{D}$ does not contain any point of $P\setminus\{a,b\}$.
\end{lemma}
\begin{proof}
  Let $e=(P_i,P_j)$ be the edge in $MST(G(\mathcal{P}))$ corresponding to $\CD{a}{b}$. Note that $w(e)=|ab|$. By contradiction, suppose that $\CD{a}{b}$ contains a point $c\in P\setminus\{a,b\}$. Three cases arise: (i) $c\in P_i$, (ii) $c\in P_j$, (iii) $c\in P_l$ where $l\neq i$ and $l\neq j$. In case (i) the edge $(c,b)$ between $c\in P_i$ and $b\in P_j$ is smaller than $(a,b)$; contradicting that $w(e)=|ab|$ in $G(\mathcal{P})$.  In case (ii) the edge $(a,c)$ between $a\in P_i$ and $c\in P_j$ is smaller than $(a,b)$; contradicting that $w(e)=|ab|$ in $G(\mathcal{P})$. In case (iii) the edge $(a,c)$ (resp. $(c,b)$) between $P_i$ and $P_l$ (resp. $P_l$ and $P_j$) is smaller than $(a,b)$; contradicting that $e$ is an edge in $MST(G(\mathcal{P}))$. 
\end{proof}

\begin{lemma}
\label{center-in-lemma}
 For each pair $D_i$ and $D_j$ of disks in $\mathcal{D}$, $D_i$ (resp. $D_j$) does not contain the center of $D_j$ (resp $D_i$).
\end{lemma}

\begin{proof}
 Let $(a_i,b_i)$ and $(a_j,b_j)$ respectively be the edges of $\mathcal{T}$ which correspond to $D_i$ and $D_j$. Let $C_i$ and $C_j$ be the circles representing the boundary of $D_i$ and $D_j$. W.l.o.g. assume that $C_j$ is the bigger circle, i.e., $|a_ib_i|<|a_jb_j|$. By contradiction, suppose that $C_j$ contains the center $c_i$ of $C_i$. Let $x$ and $y$ denote the intersections of $C_i$ and $C_j$. Let $x_i$ (resp. $x_j$) be the intersection of $C_i$ (resp. $C_j$) with the line through $y$ and $c_i$ (resp. $c_j$). Similarly, let $y_i$ (resp. $y_j$) be the intersection of $C_i$ (resp. $C_j$) with the line through $x$ and $c_i$ (resp. $c_j$). 

\begin{figure}[htb]
  \centering
  \includegraphics[width=.6\columnwidth]{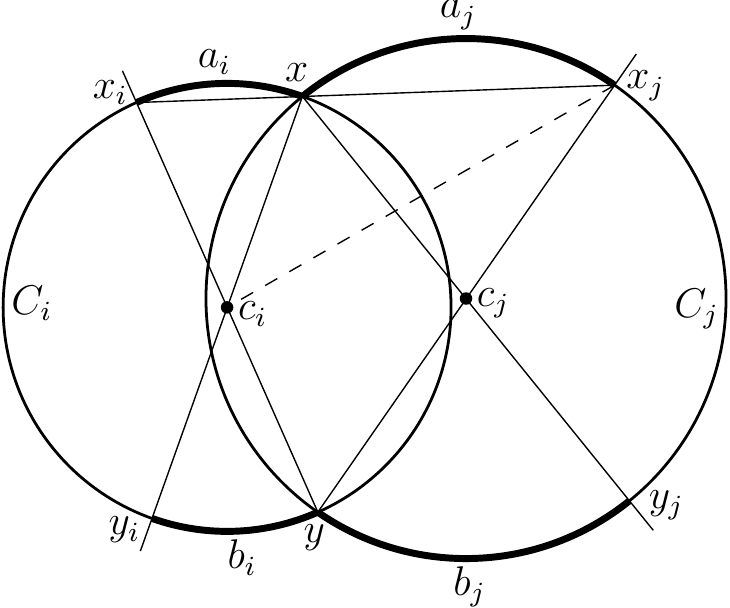}
 \caption{Illustration of Lemma~\ref{center-in-lemma}: $C_i$ and $C_j$ intersect, and $C_j$ contains the center of $C_i$.}
  \label{center-in-fig}
\end{figure}

As illustrated in Figure~\ref{center-in-fig}, the arcs $\widehat{x_ix}$, $\widehat{y_iy}$, $\widehat{x_jx}$, and $\widehat{y_jy}$ are the potential positions for the points $a_i$, $b_i$, $a_j$, and $b_j$, respectively. First we will show that the line segment $x_ix_j$ passes through $x$ and $|a_ia_j|\leq|x_ix_j|$. The angles $\angle x_ixy$ and $\angle x_jx_y$ are right angles, thus the line segment $x_ix_j$ goes through $x$. Since $\widehat{x_ix}<\pi$ (resp. $\widehat{x_jx}<\pi$), for any point $a_i\in \widehat{x_ix}, |a_ix|\leq|x_ix|$ (resp. $a_j\in \widehat{x_jx}, |a_jx|\leq|x_jx|$). Therefore, $$|a_ia_j|\leq|a_ix|+|xa_j|\leq|x_ix|+|xx_j|=|x_ix_j|.$$
Consider triangle $\bigtriangleup x_ix_jy$ which is partitioned by segment $c_ix_j$ into $t_1=\bigtriangleup x_ix_jc_i$ and $t_2=\bigtriangleup c_ix_jy$. It is easy to see that $|x_ic_i|$ in $t_1$ is equal to $|c_iy|$ in $t_2$, and the segment $c_ix_j$ is shared by $t_1$ and $t_2$. Since $c_i$ is inside $C_j$ and $\widehat{yx_j}=\pi$, the angle $\angle yc_ix_j>\frac{\pi}{2}$. Thus, $\angle x_ic_ix_j$ in $t_1$ is smaller than $\frac{\pi}{2}$ (and hence smaller than $\angle yc_ix_j$ in $t_2$). That is,  $|x_ix_j|$ in $t_1$ is smaller than $|x_jy|$ in $t_2$. Therefore,

$$|a_ia_j|\leq|x_ix_j|<|x_jy|=|a_jb_j|.$$

By symmetry $|b_ib_j|<|a_jb_j|$. Therefore $\max\{|a_ia_j|,|b_ib_j|\}<\max\{|a_ib_i|,|a_jb_j|\}$. In addition $\delta=(a_i,a_j,b_j,b_i,a_i)$ is a cycle and at least one of $(a_i,a_j)$ and $(b_i,b_j)$ does not belong to $\mathcal{T}$. This contradicts Lemma~\ref{not-mst-edge} (Note that by Observation~\ref{T-plane}, $\mathcal{T}$ is a subgraph of a minimum spanning tree of $P$).
\end{proof}

Now we show that four disks in $\mathcal{D}$ cannot intersect mutually. In other words, every point in the plane cannot lie in more than three disks in $\mathcal{D}$. In Section~\ref{proof-section} we prove the following theorem, and in Section~\ref{lower-bounds-section} we present the lower bounds on the size of a maximum matching in \kGG{k}{}.

\begin{theorem}
 \label{four-circle-theorem}
For every four disks $D_1,D_2,D_3,D_4\in\mathcal{D}$, $D_1\cap D_2\cap D_3\cap D_4=\emptyset$.
\end{theorem}
\subsection{Proof of Theorem~\ref{four-circle-theorem}}
\label{proof-section}
Let $\mathcal{X}=D_1\cap D_2\cap D_3\cap D_4$  and let $x$ be a point in $\mathcal{X}$. Let $(a_i,b_i)$ be the edge in $\mathcal{T}$ which corresponds to $D_i$, let $c_i$ be the center of $D_i$, and let $C_i$ denote the boundary of $D_i$, where $1\le i\le 4$. Denote the angle $\angle a_ixb_i$ by $\alpha_i$, where $1\le i\le 4$. Since $(a_i,b_i)$ is a diameter of $D_i$ and $x$ lies in $D_i$, $\alpha_i \ge\frac{\pi}{2}$. First we prove the following observation.
\begin{observation}
\label{inclusion-exclusion}
 For $1\le i,j\le 4$ where $i\neq j$, the angles $\alpha_i$ and $\alpha_j$ are disjoint or one is completely contained in the other.
\end{observation}
\begin{proof}
The proof is by contradiction. Suppose that $\alpha_i$ and $\alpha_j$ share some part and w.l.o.g. assume that $b_i$ is in the cone which is defined by $\alpha_j$ and $b_j$ is in the cone which is defined by $\alpha_i$. Three cases arise:
\begin{itemize}
 \item $b_i\in\bigtriangleup xa_jb_j$. In this case $b_i$ is inside $D_j$ which contradicts Lemma~\ref{D-empty}.
 \item $b_j\in\bigtriangleup xa_ib_i$. In this case $b_j$ is inside $D_i$ which contradicts Lemma~\ref{D-empty}.
 \item $b_i\notin\bigtriangleup xa_jb_j$ and $b_j\notin\bigtriangleup xa_ib_i$. In this case $(a_i,b_i)$ intersects $(a_j,b_j)$ which contradicts Observation~\ref{T-plane}.
\end{itemize}
\end{proof}

We call $\alpha_i$ a {\em blocked angle} if $\alpha_i$ is contained in an angle $\alpha_j$ where $j\neq i$, otherwise we call $\alpha_i$ a {\em free angle}.

\begin{lemma}
\label{not-all-free-angles}
At least one $\alpha_i$, where $1\le i \le 4$, is blocked.
\end{lemma}
\begin{proof}
Suppose that all angles $\alpha_i$, where $1\le i \le 4$, are free. This implies that the $\alpha_i$s are pairwise disjoint and $\alpha=\sum_{i=1}^{4}{\alpha_i} \ge 2\pi$. If $\alpha > 2\pi$, we obtain a contradiction to the fact that the sum of the disjoint angles around $x$ is at most $2\pi$. If $\alpha = 2\pi$, then the four edges $(a_i,b_i)$ where $1,\le i\le 4$, form a cycle which contradicts the fact that $\mathcal{T}$ is a subgraph of a minimum spanning tree of $P$.
\end{proof}

\begin{figure}[htb]
  \centering
\setlength{\tabcolsep}{0in}
  $\begin{tabular}{cc}
 \multicolumn{1}{m{.7\columnwidth}}{\centering\includegraphics[width=.40\columnwidth]{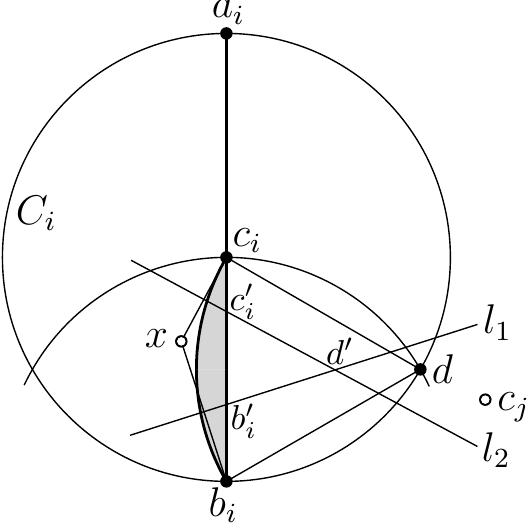}}
&\multicolumn{1}{m{.3\columnwidth}}{\centering\includegraphics[width=.06\columnwidth]{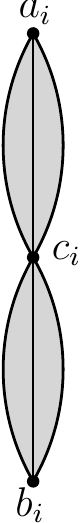}}\\
(a) & (b)
\end{tabular}$
  \caption{(a) The point $x$ should be inside the arc $\widehat{c_ib_i}$. (b) The $trap(a_i,b_i)$ which consists of two almond-shaped regions known as $trap(a_i)$ and $trap(b_i)$.}
\label{trap-fig}
\end{figure}

By Lemma~\ref{not-all-free-angles} at least one of the angles is blocked. Hereafter, assume that $\alpha_j$ is blocked by $\alpha_i$ where $1\le i,j\le 4$ and $i\neq j$. W.l.o.g. assume that $a_ib_i$ is a vertical line segment and the point $x$ (which belongs to $\mathcal{X}$) is to the left of $a_ib_i$. Thus, $a_jb_j$ and $c_j$ are to the right of $a_ib_i$. This implies that $a_ib_i\cap D_j\neq \emptyset$. See Figure~\ref{trap-fig}(a). By Lemma~\ref{center-in-lemma}, $c_i$ cannot be inside $D_j$, thus either $a_ic_i\cap D_j\neq \emptyset$ or $c_ib_i\cap D_j \neq\emptyset$, but not both. W.l.o.g. assume that $c_ib_i\cap D_j\neq \emptyset$. Let $C'$ be the circle with radius $|c_ib_i|$ which is centred at $b_i$. Let $d$ denote the intersection of $C'$ with $C_i$ which is to the right of $c_ib_i$. Consider the circle $C''$ with radius $|xb_i|$ centred at $d$. Let $\widehat{c_ib_i}$ be the closed arc of $C''$ to the left of $c_ib_i$ as shown in Figure~\ref{trap-fig}(a).

We show that $x$ cannot be outside $\widehat{c_ib_i}$. By contradiction suppose that $x$ is outside $\widehat{c_ib_i}$ (and to the left of $c_ib_i$). Let $l_1$ and $l_2$ respectively be the perpendicular bisectors of $xb_i$ and $xc_i$. Let $b'_i$ and $c'_i$ respectively be the intersection of $l_1$ and $l_2$ with $c_ib_i$ and let $d'$ be the intersection point of $l_1$ and $l_2$. Since $x$ is outside $\widehat{c_ib_i}$, the intersection point $d'$ is to the left of (the vertical line through) $d$ and inside triangle $\bigtriangleup b_ic_id$. If $c_j$ is below $l_1$ then $|c_jb_i|<|c_jx|$ and $D_j$ contains $b_i$ which contradicts Lemma~\ref{center-in-lemma}. If $c_j$ is above $l_2$ then $|c_jb_i|<|c_jx|$ and $D_j$ contains $c_i$ which contradicts Lemma~\ref{center-in-lemma}. Thus, $c_j$ is above $l_1$ and below $l_2$, and (by the initial assumption) to the right of $c_ib_i$. That is, $c_j$ is in triangle $\bigtriangleup b'_ic'_id'$. Since $\bigtriangleup b'_ic'_id'\subseteq \bigtriangleup b_ic_id\subseteq D_i$, $c_j$ lies inside $D_i$ which contradicts Lemma~\ref{center-in-lemma}. Therefore, $x$ is contained in $\widehat{c_ib_i}$. 

By symmetry $D_j$ can intersect $a_ic_i$ and/or $c_j$ can be to the left of $a_ib_i$ as well. Therefore, if $\alpha_i$ blocks $\alpha_j$, the point $x$ can be in $\widehat{c_ib_i}$ or any of the symmetric arcs. For an edge $a_ib_i$ we denote the union of these arcs by $trap(a_i,b_i)$ which is shown in Figure~\ref{trap-fig}(b). For each disk $D_i$, let $trap(D_i)= trap(a_i,b_i)$ where $(a_i, b_i)$ is the edge in $\mathcal{T}$ corresponding to $D_i$. Therefore $x$ is contained in $trap(D_i)$ which implies that $$\mathcal{X}\subseteq trap(D_i).$$ Note that $trap(D_i)$ consists of two almond-shaped symmetric regions; for simplicity we call them $trap(a_i)$ and $trap(b_i)$, i.e., $trap(D_i)=trap(a_i)\cup trap(b_i)$.

\begin{lemma}
\label{angle-in-trap}
 For any point $x\in trap(a_i,b_i)$, $\angle a_ixb_i \ge 150^\circ$.
\end{lemma}
\begin{proof}
 See Figure~\ref{trap-fig}(a). The angle $\angle b_idc_i = 60^\circ$, which implies that $\widehat{c_ib_i}= 60^\circ$. Thus, for any point $x'$ on the arc $\widehat{c_ib_i}$, $\angle x'c_ib_i + \angle x'b_ic_i = 30^\circ$, and hence for any point $x$ in $\widehat{c_ib_i}$, $\angle xc_ib_i + \angle xb_ic_i \le 30^\circ$. This implies that in $\bigtriangleup xb_ic_i$, $\angle b_ixc_i \ge 150^\circ$. On the other hand $\angle b_ixc_i \le \angle b_ixa_i$, which proves the lemma.
\end{proof}

\begin{figure}[htb]
  \centering
\setlength{\tabcolsep}{0in}
  $\begin{tabular}{ccc}
 \multicolumn{1}{m{.33\columnwidth}}{\centering\includegraphics[width=.15\columnwidth]{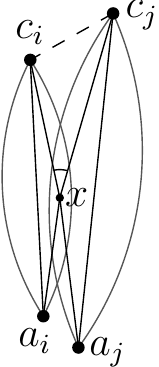}}
&\multicolumn{1}{m{.33\columnwidth}}{\centering\includegraphics[width=.15\columnwidth]{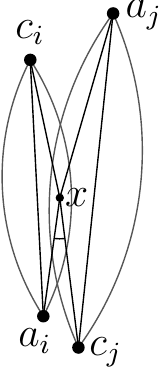}}
&\multicolumn{1}{m{.33\columnwidth}}{\centering\includegraphics[width=.14\columnwidth]{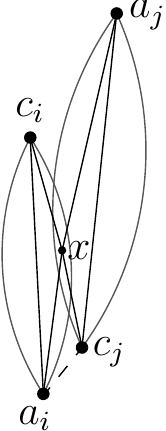}}\\
(a) & (b)& (c)
\end{tabular}$
  \caption{Illustration of Lemma~\ref{intersecting-trap}.}
\label{trap-intersection-fig}
\end{figure}

\begin{lemma}
\label{intersecting-trap}
For any two disks $D_i$ and $D_j$ in $\mathcal{D}$, $trap(D_i)\cap trap(D_j)=\emptyset$.
\end{lemma}
\begin{proof}
 We prove this lemma by contradiction. Suppose $x\in trap(D_i)\cap trap(D_j)$ and w.l.o.g. assume that $x\in trap(a_i)\cap trap(a_j)$ as shown in Figure~\ref{trap-intersection-fig}. Connect $x$ to $a_i$, $c_i$, $a_j$, and $c_j$ ($a_i$ may be identified with $a_j$). As shown in the proof of Lemma~\ref{angle-in-trap}, $\min\{\angle a_ixc_i, \angle a_jxc_j\} > 150^\circ$. Two configurations may arise: 
\begin{itemize}
 \item $\angle c_ixc_j \le 60^{\circ}$. In this case $|c_ic_j|\le \max\{|xc_i|,|xc_j|\}$. W.l.o.g. assume that $|xc_i|\le|xc_j|$ which implies that $|c_ic_j|\le |xc_j|$; see Figure ~\ref{trap-intersection-fig}(a). Clearly $|xc_j|<|c_ja_j|$, and hence $|c_ic_j|<|c_ja_j|$. Thus, $D_j$ contains $c_i$ which contradicts Lemma~\ref{center-in-lemma}.
  \item $\angle c_ixc_j > 60^{\circ}$. In this case $\angle a_ixc_j \le 60^\circ$ and $\angle a_jxc_i \le 60^\circ$, hence $|a_ic_j|\le \max\{|a_ix|,|c_jx|\}$ and $|a_jc_i|\le \max\{|a_jx|,|c_ix|\}$. Three configurations arise:

\begin{itemize}
 \item $|a_ix|<|c_jx|$, in this case $|a_ic_j|< |c_jx|<|c_ja_j|$ and hence $D_j$ contains $a_i$. See Figure~\ref{trap-intersection-fig}(b).
  \item $|a_jx|<|c_ix|$, in this case $|a_jc_i|< |c_ix|< |c_ia_i|$ and hence $D_i$ contains $a_j$. 
  \item $|a_ix|\ge|c_jx|$ and $|a_jx|\ge|c_ix|$, in this case w.l.o.g. assume that $|a_ix|\le|a_jx|$. Thus $|a_ic_j|\le |a_ix|\le |a_jx|<|a_jc_j|$ which implies that $D_j$ contains $a_i$. See Figure~\ref{trap-intersection-fig}(b).
\end{itemize}
All cases contradict Lemma~\ref{D-empty}. 
\end{itemize}
\end{proof}

Recall that each blocking angle is representing a trap. Thus, by Lemma~\ref{not-all-free-angles} and Lemma~\ref{intersecting-trap}, we have the following corollary:

\begin{corollary}
\label{one-blocked-angle}
Exactly one $\alpha_i$, where $1\le i\le 4$, is blocked.
\end{corollary}
Recall that $\alpha_j$ is blocked by $\alpha_i$, $a_ib_i$ is vertical line segment, $c_j$ is to the right of $a_ib_i$, and $x\in\widehat{c_ib_i}$. As a direct consequence of Corollary~\ref{one-blocked-angle}, $\alpha_i$, $\alpha_k$, and $\alpha_l$ are free angles, where $1\le\allowbreak i,j,\allowbreak k,\allowbreak l\le\allowbreak 4$ and $i\neq j\neq k\neq l$. In addition, $c_k$ and $c_l$ are to the left of $a_ib_i$. It is obvious that $$\mathcal{X}\subseteq trap(D_i)\cap D_k \cap D_l.$$

\begin{figure}[htb]
  \centering
  \includegraphics[width=.5\columnwidth]{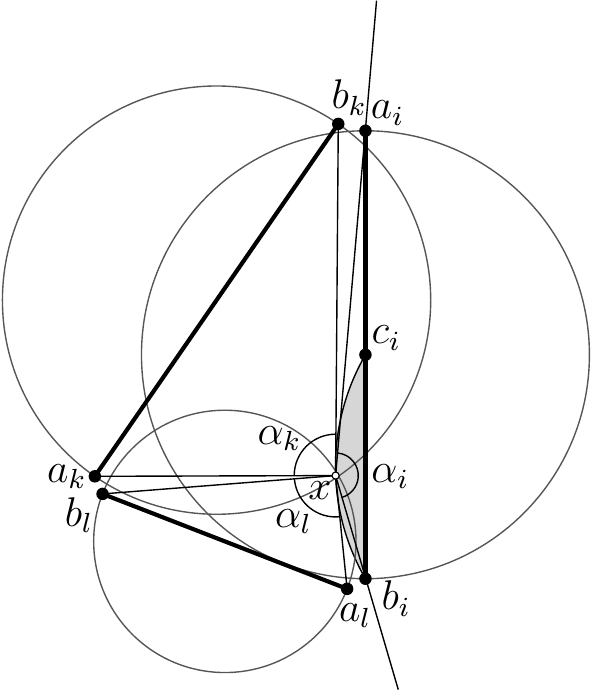}
 \caption{Illustration of Lemma~\ref{intersecting-trap2}.}
  \label{trap2-fig}
\end{figure}

\begin{lemma}
\label{intersecting-trap2}
For a blocking angle $\alpha_i$ and free angles $\alpha_k$ and $\alpha_l$, $trap(D_i)\allowbreak\cap D_k \allowbreak\cap D_l=\emptyset$. 
\end{lemma}
\begin{proof}
Since $\alpha_i$ is a blocking angle and $\alpha_k$, $\alpha_l$ are free angles, $c_k$ and $c_l$ are on the same side of $a_ib_i$.
 By contradiction, suppose that $x\in trap(D_i)\cap D_j \cap D_k$. See Figure~\ref{trap2-fig}. It is obvious that $\max\{|xa_i|, |xb_i|\}< |a_ib_i|$, $\max\{|xa_k|, |xb_k|\}< |a_kb_k|$, and $\max\{|xa_l|, |xb_l|\}<|a_lb_l|$. By Lemma~\ref{angle-in-trap}, $\alpha_i \ge 150^\circ$. In addition $\alpha_k, \alpha_l \allowbreak\ge\allowbreak 90^\circ$. Thus, $\max\{\angle a_ixb_k, \angle a_kxb_l, \angle a_lxb_i\}\allowbreak\le\allowbreak 30^\circ$. Hence, $|a_ib_k|\allowbreak<\allowbreak\max\{|xa_i|,|xb_k|\}$, $|a_kb_l|<\max\{|xa_k|,|xb_l|\}$, and $|a_lb_i|<\max\{|xa_l|,\allowbreak|xb_i|\}$. Therefore, $\max\{|a_ib_k|,|a_kb_l|, |a_lb_i|\}\allowbreak<\max\allowbreak\{|a_ib_i|,\allowbreak|a_kb_k|,|a_lb_l|\}$. In addition $\delta=(a_i,\allowbreak b_i,a_l,\allowbreak b_l,a_k,\allowbreak b_k,a_i)$ is a cycle and at least one of $(a_i,b_k)$, $(a_k,b_l)$ and $(a_l,b_i)$ does not belong to $\mathcal{T}$. This contradicts Lemma~\ref{not-mst-edge}. 
\end{proof}
Thus, $\mathcal{X}=\emptyset$; which complete the proof of Theorem~\ref{four-circle-theorem}.

\subsection{Lower Bounds}
\label{lower-bounds-section}

In this section we present some lower bounds on the size of a maximum matching in \kGG{2}{}, \kGG{1}{}, and \kGG{0}{}.

\begin{theorem}
 \label{matching-2GG}
For a set $P$ of an even number of points, \kGG{2}{} has a perfect matching.
\end{theorem}
\begin{proof}
First we show that by removing a set $S$ of $s$ points from \kGG{2}{}, at most $s+1$ components are generated. Then we show that at least one of these components must be even. Using Theorem~\ref{Tutte}, we conclude that \kGG{2}{} has a perfect matching.

Let $S$ be a set of $s$ vertices removed from \kGG{2}{}, and let $\mathcal{C}=\{C_1, \dots, C_{m(s)}\}$ be the resulting $m(s)$ components, where $m$ is a function depending on $s$. Actually $\mathcal{C}=\text{\kGG{2}{}}-S$ and $\mathcal{P}=\{V(C_1),\dots, V(C_{m(s)})\}$ is a partition of the vertices in $P\setminus S$. 

{\bf\em  Claim 1.} $m(s)\le s+1$. Let $G(\mathcal{P})$ be the complete graph with vertex set $\mathcal{P}$ which is constructed as described above. Let $\mathcal{T}$ be the set of all edges in $P$ corresponding to the edges of $MST(G(\mathcal{P}))$ and let $\mathcal{D}$ be the set of disks corresponding to the edges of $\mathcal{T}$. It is obvious that $\mathcal{T}$ contains $m(s)-1$ edges and hence $|\mathcal{D}|=m(s)-1$. Let $F=\{(p,D):p\in S, D\in \mathcal{D}, p\in D\}$ be the set of all (point, disk) pairs where $p\in S$, $D\in \mathcal{D}$, and $p$ is inside $D$. By Theorem~\ref{four-circle-theorem} each point in $S$ can be inside at most three disks in $\mathcal{D}$. Thus, $|F|\le 3\cdot|S|$.
Now we show that each disk in $\mathcal{D}$ contains at least three points of $S$ in its interior.  
Consider any disk $D\in \mathcal{D}$ and let $e=(a,b)$ be the edge of $\mathcal{T}$ corresponding to $D$. By Lemma~\ref{D-empty}, $D$ does not contain any point of $P\setminus S$. Therefore, $D$ contains at least three points of $S$, because otherwise $(a,b)$ is an edge in \kGG{2}{} which contradicts the fact that $a$ and $b$ belong to different components in $\mathcal{C}$. Thus, each disk in $\mathcal{D}$ has at least three points of $S$. That is, $3\cdot|\mathcal{D}|\le|F|$. Therefore, $3(m(s)-1)\le |F|\le 3s$, and hence $m(s)\le s+1$.

{\bf \em Claim 2}: $o(\mathcal{C})\le s$. By Claim 1, $|\mathcal{C}|=m(s)\le s+1$. If $|\mathcal{C}|\le s$, then $o(\mathcal{C})\le s$. Assume that $|\mathcal{C}|=s+1$. Since $P=S\cup \{\bigcup^{s+1}_{i=1}{V(C_i)}\}$, the total number of vertices of $P$ is equal to $n=s+\sum_{i=1}^{s+1}{|V(C_i)|}$. Consider two cases where (i) $s$ is odd, (ii) $s$ is even. In both cases if all the components in $\mathcal{C}$ are odd, then $n$ is odd; contradicting our assumption that $P$ has an even number of vertices. Thus, $\mathcal{C}$ contains at least one even component, which implies that $o(\mathcal{C})\le s$.

Finally, by Claim 2 and Theorem~\ref{Tutte}, we conclude that \kGG{2}{} has a perfect matching.
\end{proof}

\begin{theorem}
\label{matching-1GG}
For every set $P$ of $n$ points, \kGG{1}{} has a matching of size at least $\frac{2(n-1)}{5}$.
\end{theorem}

\begin{proof}
Let $S$ be a set of $s$ vertices removed from \kGG{1}{}, and let $\mathcal{C}=\{C_1, \dots, C_{m(s)}\}$ be the resulting $m(s)$ components. Actually $\mathcal{C}=\text{\kGG{1}{}}-S$ and $\mathcal{P}=\{V(C_1),\dots, V(C_{m(s)})\}$ is a partition of the vertices in $P\setminus S$. Note that $o(\mathcal{C})\le m(s)$.
Let $M^*$ be a maximum matching in \kGG{1}{}. By Theorem~\ref{Berge}, 

\begin{align}
\label{align0}
|M^*|&= \frac{1}{2}(n-\text{def}(\text{\kGG{1}{}})),
\end{align}

where

\begin{align}
\label{align1}
\text{def}(\text{\kGG{1}{}})&= \max\limits_{S\subseteq P}(o(\mathcal{C})-|S|)\nonumber\\
& \le \max\limits_{S\subseteq P}(|\mathcal{C}|-|S|)\nonumber\\
& = \max\limits_{0\le s\le n}(m(s)-s).
\end{align}
Define $G(\mathcal{P})$, $\mathcal{T}$, $\mathcal{D}$, and $F$ as in the proof of Theorem~\ref{matching-2GG}. By Theorem~\ref{four-circle-theorem}, $|F|\le 3\cdot|S|$.
By the same reasoning as in the proof of Theorem~\ref{matching-2GG}, each disk in $\mathcal{D}$ has at least two points of $S$ in its interior. Thus, $2\cdot|\mathcal{D}|\le|F|$. Therefore, $2(m(s)-1)\le |F| \le 3s$, and hence

\begin{equation}
\label{ineq1}
 m(s)\le\frac{3s}{2}+1.
\end{equation} 

In addition, $s+m(s)=|S|+|\mathcal{C}|\le |P|=n$, and hence

\begin{equation}
\label{ineq2}               
m(s)\le n-s.
\end{equation}

By Inequalities~(\ref{ineq1}) and ~(\ref{ineq2}), 

\begin{equation}
\label{ineq3}               
m(s)\le \min\{\frac{3s}{2}+1, n-s\}.
\end{equation}

Thus, by (\ref{align1}) and (\ref{ineq3})

\begin{align}
\label{align2}
\text{def}(\text{\kGG{1}{}})&\le \max\limits_{0\le s\le n}(m(s)-s)\nonumber\\
&\le \max\limits_{0\le s\le n}\{\min\{\frac{3s}{2}+1, n-s\}-s\}\nonumber\\
&= \max\limits_{0\le s\le n}\{\min\{\frac{s}{2}+1, n-2s\}\}\nonumber\\
&= \frac{n+4}{5},
\end{align}

where the last equation is achieved by setting $\frac{s}{2}+1$ equal to $n-2s$, which implies $s=\frac{2(n-1)}{5}$. Finally by substituting (\ref{align2}) in Equation (\ref{align0}) we have
$$
|M^*|\ge \frac{2(n-1)}{5}.
$$
\end{proof}

By similar reasoning as in the proof of Theorem~\ref{matching-1GG} we have the following Theorem.

\begin{figure}[htb]
  \centering
  \includegraphics[width=.6\columnwidth]{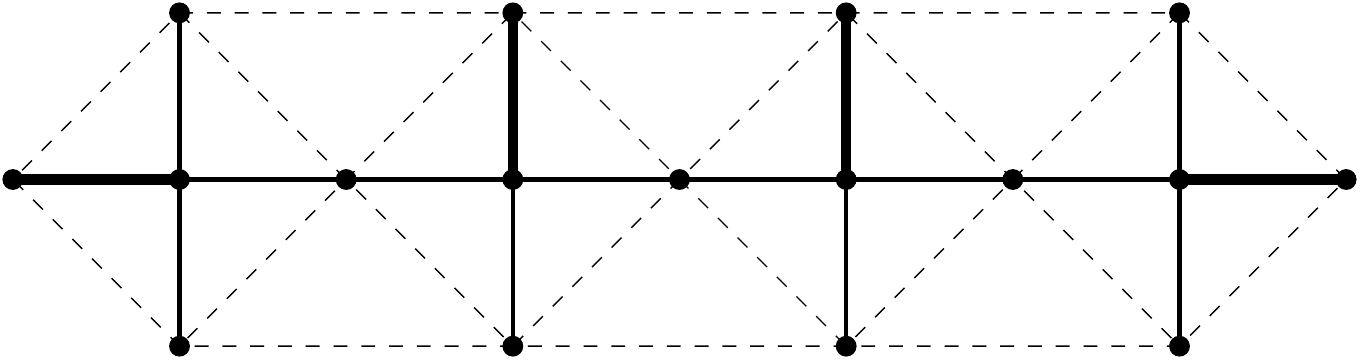}
 \caption{A \kGG{0}{} of  $n = 17$ points with a maximum matching of size $\frac{n-1}{4}=4$ (bold edges). The dashed edges do not belong to the graph because any of their corresponding closed disks has a point on its boundary.}
  \label{tight-0GG}
\end{figure}

\begin{theorem}
\label{matching-0GG}
For every set $P$ of $n$ points, \kGG{0}{} has a matching of size at least $\frac{n-1}{4}$.
\end{theorem}

The bound in Theorem~\ref{matching-0GG} is tight, as can be seen from the graph in Figure~\ref{tight-0GG}, for which the maximum matching has size $\frac{n-1}{4}$. Actually this is a Gabriel graph of maximum degree four which is a tree. The dashed edges do not belong to \kGG{0}{} because any closed disk which has one of these edges as diameter has a point on its boundary. Observe that each edge in any matching is adjacent to one of the vertices of degree four.

\begin{paragraph}{Note:}For a point set $P$, let $\nu_k(P)$ and $\alpha_k(P)$ respectively denote the size of a maximum matching and a maximum independent set in \kGG{k}{}. For every edge in the maximum matching, at most one of its endpoints can be in the maximum independent set. Thus,$$\alpha_k(P)\le |P| - \nu_k(P).$$
By combining this formula with the results of Theorems ~\ref{matching-0GG}, \ref{matching-1GG}, \ref{matching-2GG}, respectively, we have $\alpha_0(P)\le \frac{3n+1}{4}$, $\alpha_1(P)\le \frac{3n+2}{5}$, and $\alpha_2(P)\le \lceil\frac{n}{2}\rceil$. The \kGG{0}{} graph in Figure~\ref{tight-0GG} has an independent set of size $\frac{3n+1}{4}=13$, which shows that this bound is tight for \kGG{0}{}. On the other hand, \kGG{0}{} is planar and every planar graph is 4-colorable; which implies that $\alpha_0(P)\ge \lceil\frac{n}{4}\rceil$. There are some examples of \kGG{0}{} in \cite{Matula1980} such that $\alpha_0(P)= \lceil\frac{n}{4}\rceil$, which means that this bound is tight as well.
\end{paragraph}

\section{Blocking Higher-Order Gabriel Graphs}
\label{blocking-section}
In this section we consider the problem of blocking higher-order Gabriel graphs. Recall that a point set $K$ blocks \kGG{k}{(P)} if in \kGG{k}{(P\cup K)} there is no edge connecting two points in $P$. 

\begin{theorem}
\label{blocking-thr1}
For every set $P$ of $n$ points, at least $\lceil\frac{n-1}{3}\rceil$ points are necessary to block \kGG{0}{(P)}.
\end{theorem}
\begin{proof}
Let $K$ be a set of $m$ points which blocks \kGG{0}{(P)}. Let $G(\mathcal{P})$ be the complete graph with vertex set $\mathcal{P}=P$. Let $\mathcal{T}$ be a minimum spanning tree of $G(\mathcal{P})$ and let $\mathcal{D}$ be the set of closed disks corresponding to the edges of $\mathcal{T}$. It is obvious that $|\mathcal{D}|=n-1$. By Lemma~\ref{D-empty} each disk $\CD{a}{b}\in\mathcal{D}$ does not contain any point of $P\setminus\{a,b\}$, thus,  $\mathcal{T}\subseteq\text{\kGG{0}{(P)}}$. To block each edge of $\mathcal{T}$, corresponding to a disk in $\mathcal{D}$, at least one point is necessary. By Theorem~\ref{four-circle-theorem} each point in $K$ can lie in at most three disks of $\mathcal{D}$. Therefore, $m\ge\lceil\frac{n-1}{3}\rceil$, which implies that at least $\lceil\frac{n-1}{3}\rceil$ points are necessary to block all the edges of $\mathcal{T}$ and hence \kGG{0}{(P)}.
\end{proof}
\begin{figure}[htb]
  \centering
\setlength{\tabcolsep}{0in}
  $\begin{tabular}{cc}
 \multicolumn{1}{m{.75\columnwidth}}{\centering\includegraphics[width=.65\columnwidth]{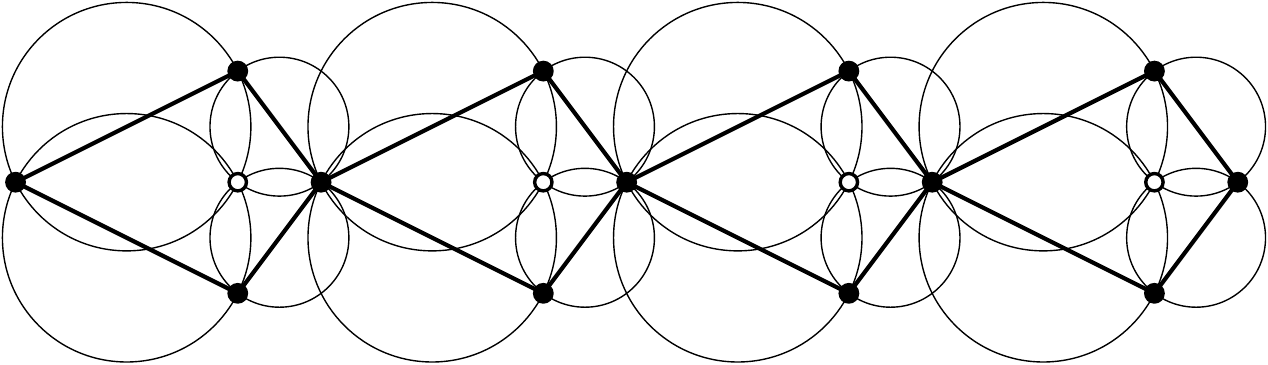}}
&\multicolumn{1}{m{.25\columnwidth}}{\centering\includegraphics[width=.23\columnwidth]{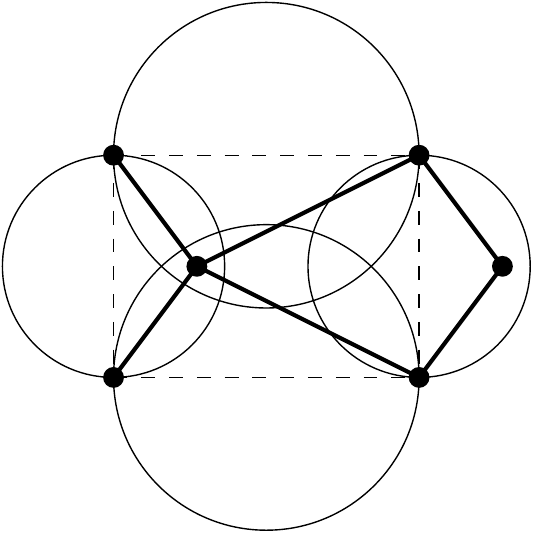}}\\
(a)&(b)
\end{tabular}$
  \caption{(a) \kGG{0}{} graph of $n=13$ points (in bold edges) which is blocked by $\lceil\frac{n-1}{3}\rceil=4$ white points, (b) dashed edges do not belomg to \kGG{0}{}.}
\label{blocking-fig}
\end{figure}

Figure~\ref{blocking-fig}(a) shows a \kGG{0}{} with $n=13$ (black) points which is blocked by $\lceil\frac{n-1}{3}\rceil=4$ (white) points. Note that all the disks, corresponding to the edges of every cycle, intersect at the same point in the plane (where we have placed the white points). As shown in Figure~\ref{blocking-fig}(b), the dashed edges do not belong to \kGG{0}{}. Thus, the lower bound provided by Theorem~\ref{blocking-thr1} is tight. It is easy to generalize the result of Theorem~\ref{blocking-thr1} to higher-order Gabriel graphs. Since in a \kGG{k}{} we need at least $k+1$ points to block an edge of $\mathcal{T}$ and each point can be inside at most three disks in $\mathcal{D}$, we have the following corollary:

\begin{corollary}
For every set $P$ of $n$ points, at least $\lceil\frac{(k+1)(n-1)}{3}\rceil$ points are necessary to block \kGG{k}{(P)}.
\end{corollary}

In \cite{Aronov2013} the authors showed that every Gabriel graph can be blocked by a set $K$ of $n-1$ points by putting a point slightly to the right of each point of $P$, except for the rightmost one. Every disk with diameter determined by two points of $P$ will contain a point of $K$. Using a similar argument one can block a \kGG{k}{} by putting $k+1$ points slightly to the right of each point of $P$, except for the rightmost one. Thus,

\begin{corollary}
 For every set $P$ of $n$ points, there exists a set of $(k+1)(n-1)$ points that blocks \kGG{k}{(P)}.
\end{corollary}

Note that this upper bound is tight, because if the points of $P$ are on a line, the disks representing the minimum spanning tree are disjoint and each disk needs $k+1$ points to block the corresponding edge.

\section{Conclusion}
\label{conclusion}
In this paper, we considered the bottleneck and perfect matching admissibility of higher-order Gabriel graphs. We proved that
\begin{itemize}
  \item \kGG{10}{} contains a Euclidean bottleneck matching of $P$ and \kGG{8}{} may not have any.
  \item \kGG{0}{} has a matching of size at least $\frac{n-1}{4}$ and this bound is tight.
  \item \kGG{1}{} has a matching of size at least $\frac{2(n-1)}{5}$.
  \item \kGG{2}{} has a perfect matching.
    \item $\lceil\frac{n-1}{3}\rceil$ points are necessary to block \kGG{0}{} and this bound is tight.
  \item $\lceil\frac{(k+1)(n-1)}{3}\rceil$ points are necessary and $(k+1)(n-1)$ points are sufficient to block \kGG{k}{}.
\end{itemize}
We leave a number of open problems:
\begin{itemize}
  \item Does \kGG{9}{} contain a Euclidean bottleneck matching of $P$?
  \item What is a tight lower bound on the size of a maximum matching in \kGG{1}{}?
\end{itemize}

\bibliographystyle{abbrv}
\bibliography{GG-matching.bib}

\begin{thebibliography}{10}

\bibitem{Abellanas2009}
M.~Abellanas, P.~Bose, J.~Garc\'{\i}a-L{\'o}pez, F.~Hurtado, C.~M. Nicol{\'a}s,
  and P.~Ramos.
\newblock On structural and graph theoretic properties of higher order
  {D}elaunay graphs.
\newblock {\em Int. J. Comput. Geometry Appl.}, 19(6):595--615, 2009.

\bibitem{Aichholzer2013}
O.~Aichholzer, R.~F. Monroy, T.~Hackl, M.~J. van Kreveld, A.~Pilz, P.~Ramos,
  and B.~Vogtenhuber.
\newblock Blocking {D}elaunay triangulations.
\newblock {\em Comput. Geom.}, 46(2):154--159, 2013.

\bibitem{Aronov2013}
B.~Aronov, M.~Dulieu, and F.~Hurtado.
\newblock Witness {G}abriel graphs.
\newblock {\em Comput. Geom.}, 46(7):894--908, 2013.

\bibitem{Babu2013}
J.~Babu, A.~Biniaz, A.~Maheshwari, and M.~Smid.
\newblock Fixed-orientation equilateral triangle matching of point sets.
\newblock To appear in {\em Theoretical Computer Science}.

\bibitem{Berge1958}
C.~Berge.
\newblock Sur le couplage maximum d'un graphe.
\newblock {\em C. R. Acad. Sci. Paris}, 247:258--259, 1958.

\bibitem{Biniaz2014}
A.~Biniaz, A.~Maheshwari, and M.~Smid.
\newblock Higher-order triangular-distance {D}elaunay graphs: Graph-theoretical
  properties.
\newblock arXiv: 1409.5466, 2014.

\bibitem{Bose2013}
P.~Bose, S.~Collette, F.~Hurtado, M.~Korman, S.~Langerman, V.~Sacristan, and
  M.~Saumell.
\newblock Some properties of $k$-{D}elaunay and $k$-{G}abriel graphs.
\newblock {\em Comput. Geom.}, 46(2):131--139, 2013.

\bibitem{Chang1991}
M.-S. Chang, C.~Y. Tang, and R.~C.~T. Lee.
\newblock 20-relative neighborhood graphs are {H}amiltonian.
\newblock {\em Journal of Graph Theory}, 15(5):543--557, 1991.

\bibitem{Chang1992}
M.-S. Chang, C.~Y. Tang, and R.~C.~T. Lee.
\newblock Solving the {E}uclidean bottleneck matching problem by $k$-relative
  neighborhood graphs.
\newblock {\em Algorithmica}, 8(3):177--194, 1992.

\bibitem{Dillencourt1990}
M.~B. Dillencourt.
\newblock Toughness and {D}elaunay triangulations.
\newblock {\em Discrete {\&} Computational Geometry}, 5:575--601, 1990.

\bibitem{Gabriel1969}
K.~R. Gabriel and R.~R. Sokal.
\newblock A new statistical approach to geographic variation analysis.
\newblock {\em Systematic Zoology}, 18(3):259--278, 1969.

\bibitem{Matula1980}
D.~W. Matula and R.~R. Sokal.
\newblock Properties of {G}abriel graphs relevant to geographic variation
  research and the clustering of points in the plane.
\newblock {\em Geographical Analysis}, 12(3):205--222, 1980.

\bibitem{Su1990}
T.-H. Su and R.-C. Chang.
\newblock The $k$-{G}abriel graphs and their applications.
\newblock In {\em SIGAL International Symposium on Algorithms}, pages 66--75,
  1990.

\bibitem{Tutte1947}
W.~T. Tutte.
\newblock The factorization of linear graphs.
\newblock {\em Journal of the {L}ondon Mathematical Society}, 22(2):107--111,
  1947.

\end{thebibliography}

\end{document}